\documentclass[11pt]{article}

\usepackage{amssymb,amsthm,amsmath}
\usepackage[mathscr]{eucal}

\newcommand{\Irr}{{\rm Irr}}
\newcommand{\wt}{\mbox{wt}}

\usepackage{bm} 

\newcommand{\C}{\mbox{$\mathbb C$}}

\newcommand{\Z}{\mbox{$\mathbb Z$}}

\newcommand{\vspan}{{\rm span}}

\newcommand\inner[1]{\langle \, #1  \, \rangle}
\newcommand{\cart}[3]{\ensuremath{{#1}_{#2} \times \cdots \times {#1}_{#3}}}
\newcommand{\ocart}[3]{\ensuremath{{#1}_{#2} \otimes \cdots \otimes {#1}_{#3}}}

\newcommand{\pperp}{\ensuremath{{\rm \perp}\kern-.60em {\rm \perp} }}
\newcommand{\npperp}{\ensuremath{{\rm \perp}\kern-.60em {\rm \not\perp} }}
\newcommand\set[1]{\{1, \ldots, #1 \}}
\newcommand\norm[1]{\| #1 \|}

\newtheorem{theorem}{Theorem}[section]
\newtheorem{lemma}[theorem]{Lemma}
\newtheorem{proposition}[theorem]{Proposition}
\newtheorem{example}[theorem]{Example}

\begin{document}
\author{Jay H.~Beder
\\ Jeb F.~Willenbring
\\Department of Mathematical Sciences\\University of Wisconsin-Milwaukee\\
P.O. Box 413\\Milwaukee, WI 53201-0413\\
{\tt beder@uwm.edu, jw@uwm.edu}}

\title{Invariance of generalized wordlength patterns}
\date{}

\maketitle

\vspace{-1cm}

\begin{abstract}
The generalized wordlength pattern (GWLP) introduced by Xu and Wu (2001) for an
arbitrary fractional factorial design allows one to extend the use of the
minimum aberration criterion to such designs. Ai and Zhang (2004) defined the
$J$-characteristics of a design and showed that they uniquely determine the
design.  While both the GWLP and the $J$-characteristics require indexing the
levels of each factor by a cyclic group, we see that the definitions carry over
with appropriate changes if instead one uses an arbitrary abelian group. This
means that the original definitions rest on an arbitrary choice of group
structure. We show that the GWLP of a design is independent of this choice, but
that the $J$-characteristics are not.  We briefly discuss some implications of
these results.

\end{abstract}


{\footnotesize {\bf AMS(MOS) subject classification.}
Primary: 62K15; 
Secondary:
05B15, 
15A69, 
20C15, 
62K05 
}

\section{Introduction}\label{intro}
In a regular fractional factorial design $D$, the quantities
  \[ A_i(D) = \mbox{the number of defining words of length} \; i \]
contain useful information about the design. In
particular, the smallest index $i$ for which $A_i(D) > 0$ is the resolution of
the design. Moreover, one way of comparing two designs having $k$ factors and
equal resolution is to compare their \emph{wordlength patterns} $(A_1, A_2,
\ldots, A_k)$ \cite{Franklin84,FriesHunter80}. The better design is said to
have less \emph{aberration}.

While nonregular designs no longer have defining words as such, a
\emph{generalized wordlength pattern (GWLP)} can be defined for them
combinatorially.  This was done for two-level designs by Tang and Deng
\cite{TangDeng99}, and was generalized to arbitrary (possibly mixed-level)
designs by Xu and Wu \cite{XuWu01} using group characters.

An intermediate computation in the two-level case gives a set of values that
Tang and Deng called $J$-characteristics (first introduced in
\cite{DengTang99}), and Tang \cite{Tang01} showed that these numbers completely
determine the design $D$, somewhat analogous to the way that a defining
subgroup determines a regular design.  Ai and Zhang \cite{AiZhang04}
generalized this to arbitrary designs by looking closely at the corresponding
computation in \cite{XuWu01}.

In defining generalized wordlength patterns of arbitrary designs, Xu and Wu
\cite{XuWu01} assigned to the $i$th factor the cyclic group $\Z_{s_i}$, where
$s_i$ = the number of levels of the factor.  While this choice is a
computational convenience, it is also arbitrary, and in fact the calculation of
the GWLP can be carried through using other abelian groups as well, as we
indicate below.

This, however, raises the following question for non-prime $s_i$. Since the
(irreducible) characters of two groups of equal order will generally be
different, does the choice of group affect either the $J$-characteristics or
the GWLP of a given design? Certainly any dependence of the GWLP on an
arbitrary choice would raise a serious question about its use in comparing
designs using relative aberration. It will be clearly seen that the
$J$-characteristics do depend on this choice. However, perhaps surprisingly,
this does not affect the values of the GWLP. That is our main result.

There are many excellent expositions of character theory, such as
\cite{Isaacs76}, \cite{Ledermann87} and~\cite{Serre77}. In general we will
mention known results without citation.  We will also use a number of results
from multilinear algebra (the theory of tensor products).  These are collected
in an appendix.  \vspace{1cm}

\textbf{Notation}. We will denote the integers by \Z, and the integers modulo
$s$ by $\Z_s$ as above.  The complex numbers will be denoted by \C\ and complex
Euclidean space by $\C^s$. Vectors in $\C^s$ will be viewed as columns. The
conjugate of $z \in \C$ will be denoted by $\bar{z}$, the transpose of a vector
or matrix by a prime ($'$), and the adjoint (or conjugate transpose) of a
matrix or linear transformation $A$ by $A^*$. The inner product of $v = [v_1,
\ldots, v_s]'$ and $w = [w_1, \ldots, w_s]' \in \C^s$ is given by
  \begin{equation} \label{innerprod-eq} \inner{v,w} = \sum_{i=1}^s v_i
  \overline{w}_i. \end{equation}
The cardinality of a set $E$ will be written $|E|$.

The \emph{Hamming weight} of $u = (u_1, \ldots, u_k)$, $\wt(u)$, is the number
of nonzero components of $u$. (In Section~\ref{definitions} we will replace
``nonzero" by ``nonidentity" in order to deal with groups whose identity
element is not 0.)

We alert the reader to the fact that we will use $G$ (or $G_i$) as an index
set, with elements $g$ or $h$.  Sometimes such sets will be groups, but often
they will be viewed just as sets.  We will try to make absolutely clear from
context when a result requires a group structure and when it doesn't.

\section{Definitions}\label{definitions}
A \emph{fractional factorial design} on $k$ factors is a multisubset $D$ of a
finite Cartesian product $G = \cart{G}{1}{k}$, that is, the set
$G$ with the element $g$ repeated $O(g)$ times, $O(g) \ge 0$.
The set $G_i$ indexes the $s_i$ levels of factor $i$, and we let $s = s_1 \cdots s_k$. We will refer to $O$ as the \emph{counting} or \emph{multiplicity function} of $D$. The elements ($k$-tuples) of the design are referred to as \emph{runs}, and the number of runs in the design, counting multiplicities, is
  \begin{equation} \label{|D|-eq}
  N = |D| = \sum_{g \in G} O(g).
  \end{equation}
The design $D$ may also be viewed as an \emph{orthogonal array}, particularly
if its runs are displayed in matrix form, say as columns of a $k \times N$
matrix.

In \cite{XuWu01} Xu and Wu defined the generalized wordlength pattern $(A_1(D),
\ldots, A_k(D))$ of $D$ as follows. If $G_i$ has $s_i$ elements, we take $G_i =
\Z_{s_i}$, the additive group of integers modulo $s_i$. This makes $G$ an
abelian group. To each $g \in \Z_s$ we associate a function $\chi_g: \Z_s \to
\C$ such that
\begin{equation} \label{cyclic_char-eq}\chi_g(h) = \xi^{gh}, \end{equation}
where $\xi$ is a primitive $s$th root of unity (say $\xi = e^{2\pi i/s}$).  For
elements $g = (g_1, \ldots, g_k)$ and $h = (h_1, \ldots, h_k)$ of $G =
\cart{G}{1}{k}$, we let
  \begin{equation} \label{char_prod-eq} \chi_{g}(h) = \prod_i \chi_{g_i}(h_i),
  \end{equation}
and define\footnote{In \cite{XuWu01} $\chi_g(D)$ is defined as $\sum_{h \in D} \chi_g(h)$, and it is to be understood that the $h$th term is repeated the number of times $h$ appears in the design \cite{Xu_personal}.  Equation (\ref{chi_D-eq}), which is essentially the same as that used in \cite{AiZhang04}, makes this explicit.}
  \begin{equation} \label{chi_D-eq} \chi_{g}(D) = \sum_{h \in G}
  O(h)\chi_{g}(h).
  \end{equation}
Finally, the ``generalized wordlengths" are given by
  \begin{equation} \label{genwordlength}
  A_j(D) = N^{-2} \sum_{\wt(g) = j} |\chi_{g}(D)|^2 \qquad \mbox{for} \; j =  1, \ldots, k, \end{equation}
where $\wt(g)$ is the Hamming weight of $g$.

Ai and Zhang \cite{AiZhang04} note that when $s_1 = \cdots = s_k = 2$ the
quantities $\chi_{g}(D)$ are the \emph{$J$-characteristics} of Tang and Deng,
and rename them so in the general case, with the notation $J_{g}(D)$.

We now indicate the way in which other groups may be used in (\ref{chi_D-eq}) and (\ref{genwordlength}).

The functions $\chi_{g_i}$ are the \emph{irreducible characters}
of the group $\Z_{s_i}$, and so the functions $\chi_{g}$ are the irreducible
characters of $G$.  Among these is $\chi_{e} \equiv 1$, the \emph{trivial
character} of $G$, corresponding to the identity $e$ of $G$.  Something similar holds for abelian groups, in particular the indexing of irreducible characters by group elements.

Specifically, the irreducible characters of an abelian group $G$ are precisely
the homomorphisms of $G$ into the multiplicative group $\C^* = \C \setminus
{0}$. The indexing of these characters is based on the following result.
\begin{theorem}  \label{IrrG-thm}
Let $\Irr(G)$ denote the set of irreducible characters of the group $G$. If $G$
is abelian, then $\Irr(G)$ forms a group under pointwise multiplication, and if
$G$ is also finite, then $G\cong \Irr(G)$. In particular, the identity element
of $G$ corresponds to the trivial character of $G$.
\end{theorem}
The isomorphism is not canonical -- and, in particular, not unique -- as it
depends on the representation of an abelian group as a product of cyclic groups
(the Fundamental Theorem of Abelian Groups), and for cyclic groups on the
choice of root of unity in (\ref{cyclic_char-eq}). (See, e.g., \cite[Theorem
2.4]{Ledermann87}).
We will assume that we have fixed an isomorphism $G_i \to \Irr(G_i)$ for each
$i$, and thus an indexing of the irreducible characters of $G_i$ by group
elements.  We will not need to know the indexing explicitly.  The irreducible
characters of the direct product $G$ are still given by (\ref{char_prod-eq}).

$J$-characteristics and generalized wordlength counts are still defined by
(\ref{chi_D-eq}) and (\ref{genwordlength}), respectively, where we now define
the \emph{weight} of the element $g = (g_1, \ldots, g_k) \in G$ to be the
number of \emph{nonidentity} components of $g$. Our main result is this:

\begin{theorem} \label{indep-thm} The quantities $A_j(D)$ in (\ref{genwordlength}) are
independent of the group structure of $G$. \end{theorem}

The proof of this theorem is given in Section \ref{independence}. Before
considering this, we take a moment to study the effect of the choice of group
on the $J$-characteristics of a design.

\section{$J$-characteristics. The character table.} \label{J-char-sec}
We see that the irreducible characters of a finite abelian group $G$ of order
$s$ may be written $\chi_{g_1}, \ldots, \chi_{g_s}$, where $g_i$ are the
elements of $G$ in some order.  The values $\chi_{g_i}(g_j)$ form the
\emph{character table} of $G$, the columns of which are mutually orthogonal and
of norm $\sqrt{s}$ (with respect to the inner product (\ref{innerprod-eq})).
Another way to say this is that the $s \times s$ matrix $H$ formed by this
table has the property that $H^*H = HH^* = sI$, where $H^*$ is the adjoint of
$H$ ($H$ is thus a complex Hadamard matrix).

Let $G = \cart{G}{1}{k}$ where each $G_i$ is an abelian group, so that $G$ is
as well, and assume that the elements of $G$ are ordered in some fashion.  (Ai
and Zhang \cite{AiZhang04} use a lexicographic or \emph{Yates} order).)  If we
consider the set of $J$-characteristics $\chi_g(D)$ and the counts $O(g)$ as $s
\times 1$ vectors $\chi$ and $O$ indexed by $g \in G$, then (\ref{chi_D-eq})
may be written
   \begin{equation} \label{chi=HO-eq}  \chi = HO. \end{equation}
Multiplying through by $H^*$, we see that $H^*\chi = H^*HO = sO,$ so that \[O =
(1/s)H^*\chi,\] and in particular that the $J$-characteristics determine the
design.  This is Theorem 1 of \cite{AiZhang04}.

However, in general $H$ depends on the group structure of $G$, and so from
(\ref{chi=HO-eq}) or directly from ({\ref{chi_D-eq}) we see that \emph{the
values of the $J$-characteristics depend on the choice of group structure}.  This is illustrated with the following example.

\begin{example} \rm Consider the 3-factor design
\begin{equation}  \label{O4-eq}
D = \left[ \begin{array}{cccc}
    0000 & aaaa & bbbb & cccc \\
    0abc & 0abc & 0abc & 0abc \\
    0abc & b0ca & ac0b & cba0
\end{array}\right].
    \end{equation}
Each factor has 4 levels, namely $0, a, b$, and $c$, and each column is a
treatment combination. One can check that this is an orthogonal array of
strength 2 and index 1 (it is taken from \cite{DieBed02}, where it is shown to
be non-regular).

For each factor the symbol set $G_i = \{0,a,b,c\}$ may be given two group
structures, namely that of the cyclic group $\Z_4$ and that of the ``Klein
4-group" $V$ (isomorphic to $\Z_2 \times \Z_2$).  Table~\ref{O4-table} displays
the non-zero values of $\chi_g(D)$ as $g$ runs over the 64 elements of $G = G_1
\times G_2 \times G_3$, where the groups $G_i$ are all $\Z_4$ or all $V$.  For example,
$\chi_{aaa}(D) = -6-2i$ using $\Z_4$, but = 8 using $V$.  Thus we see that the
values of the $J$-characteristics depend on the group structure.   Note that for both group structures we have
$\chi_g(D) = 0$ if $\wt(g) = 1$ or 2.  Such group elements have been omitted
from the table for convenience.

It is not hard to calculate the values $A_j(D)$ given by equation
(\ref{genwordlength}), where we have $N = 16$.  (The computation is shortened
by the fact that $|\pm a \pm bi|^2 = a^2 + b^2$.)  We find that under both
group structures we have $A_1(D) = A_2(D) = 0$ and $A_3(D) = 3$, as guaranteed
by Theorem~\ref{indep-thm}.

\setlength{\tabcolsep}{1.3mm}
\begin{table}\parbox{15cm}{\caption{$J$-Characteristics for Design $D$ under two different group structures
 \newline For the design in
(\ref{O4-eq}), the value of $\chi_g(D)$ is given for each $g \in G$, where $G
=$ either $\Z_4 \times \Z_4 \times \Z_4$ or $V \times V \times V$. Those $g$
for which both values of $\chi_g(D)=0$ are omitted. (Computation was done in Maple.) }}\label{O4-table}
\begin{tabular}{ccccccccccc} \\ \hline
$g=$   &$000$  &$aaa$   &$aab$   &$aac$  &$aba$   &$abb$   &$abc$  &$aca$
&$acb$   &$acc$\\ \hline
$\Z_4$ & 16 &  $-6-2i$&    $4i$&  $6-2i$& $-4i$& $4+4i$&   $-4$&   $6-2i$& 4&  $6+2i$ \\
$V$ &  16& 8&  8&  0&  $-8$&  8&  0&   0&  0&  0\\ \\ \hline

$g=$    &$baa$   &$bab$   &$bac$  &$bba$   &$bbb$   &$bbc$  &$bca$ &$bcb$
&$bcc$\\ \hline
$\Z_4$ & $4i$& $-4-4i$& 4& $4+4i$&    8& $4-4i$&  4& $-4+4i$&   $-4i$\\
$V$ &   8& $-8$&  0&     8&  8&  0&    0&  0&  0\\ \\ \hline

$g=$       &$caa$   &$cab$   &$cac$
      &$cba$   &$cbb$   &$cbc$ &$cca$   &$ccb$   &$ccc$\\ \hline
$\Z_4$ &    $6-2i$&     4&  $6+2i$&
      $-4$& $4-4i$&   $4i$&    $6+2i$&   $-4i$& $-6+2i$\\
$V$ &   0&  0&  0&  0&  0&  0&  0&  0& 16 \\ \hline
\end{tabular}

\end{table}

\end{example}

Before we leave this topic, we develop the properties of the character table a
little further.

In enumerating the elements of a group $G$ one typically chooses $g_1 =$ the
identity.  With this convention, which we shall adopt, $\chi_{g_1}$ is the
trivial character of $G$, so that $\chi_{g_1}(h) = 1$ for all $h \in G$.  On
the other hand, since $G$ is abelian, $\chi_g(g_1) = 1$ for every $g \in G$,
and so we see that $H$ must have the form
  \[ H = \left(
           \begin{array}{cccc}
             1 & 1 & \cdots & 1 \\
             1 & * & \cdots & * \\
             \vdots & \vdots &  & \vdots \\
             1 & * & \cdots & * \\
           \end{array}
         \right).
  \]

The matrix $U =  (1/\sqrt{s})H$ is said to be the \emph{normalized character
table} of $G$.  We list its important properties here, which follow from the
preceding.
  \begin{lemma} \label{unitary-lemma} $U$ unitary ($U^*U = UU^* = I$),
  and in particular defines an isometry on $\C^s$ ($\inner{Uv,Uw} = \inner{v,w}$).
  If $e = [1, 0, \ldots, 0]'$ and $b = (1/\sqrt{s})[1, \ldots, 1]'$ then
  \[ Ue = b = U^*e. \]
\end{lemma}

\section{Independence of group structure} \label{independence}
By imposing a group structure on the set $G = \cart{G}{1}{k}$, we define the
irreducible characters $\chi_g$.  We want to show that the numbers
  \[ \sum_{\wt(g) = j} |\chi_g(D)|^2, \quad j = 1, \ldots, k, \]
appearing in (\ref{genwordlength}) are independent of the group structure
chosen.  This sum is somewhat unwieldy, and so we will break it into smaller
sums over elements $g$ which are not only of weight $j$ but also differ from
the identity in exactly the same components.

To begin with, we fix an order of the elements in each set $G_i$, with the
understanding that \emph{whenever we impose a group structure, the first
element will be the identity of the group}.  We may denote by $1_i$ the chosen
element of $G_i$.

Now, for each $J \subset \set{k}$ with $|J|=j$, let
  \begin{equation} \label{SI-eq} S_J = \{g = (g_1, \cdots, g_k) \in G:
   g_i \neq 1_i \;\mbox{iff}\; i \in J \}.
  \end{equation}
(Here $J$ is merely an index set and has no relation to the $J$-characteristics
mentioned earlier.) Clearly the sets $S_J$ are disjoint and their union is the
set of elements of $G$ of weight $j$. Then
  \[ \sum_{\wt(g) = j} |\chi_g(D)|^2 = \sum_{|J| = j} \; \sum_{g \in
  S_J} |\chi_g(D)|^2. \]
We will show that for each $J$ the inner sum
  \begin{equation} \label{innersum-eq} \sum_{g \in
  S_J} |\chi_g(D)|^2 \end{equation}
is independent of the group structure chosen.

To do this, we will write these sums as squared norms of elements in an
appropriate subspace $V_J$ of $\C^s$. Assuming a fixed ordering of the elements
of $G$, the components of a vector $v \in \C^s$ are complex numbers indexed by
the elements of $G$, something like
  \begin{equation} \label{vg-vector} v = [\cdots, v_g, \cdots]'. \end{equation}
The standard basis elements are of form
  \begin{equation} \label{unit-vector} e_g = [0, \ldots, 0, 1, 0 \ldots, 0]', \end{equation}
where 1 occurs in just the $g$-th coordinate.  Then
  \begin{equation} \label{v=sum-eq} v = \sum_{g \in G} v_g e_g. \end{equation}
Let
  \[ V_J = \{v = [\cdots, v_g, \cdots]' \in \C^s: v_g = 0 \; \mbox{if} \; g \notin S_J\}. \]
It is clear that $\dim V_J = |S_J| = \prod_{i \in J} (s_i - 1)$, and that the
sum (\ref{innersum-eq}) is $\norm{M_J(\chi)}^2$ where $M_J$ is the orthogonal
projection of $\C^s$ onto $V_J$.  Now the next result follows immediately from
(\ref{chi=HO-eq}) and the fact that $H = \sqrt{s} \, U$.

\begin{proposition}  \label{MJUO-prop}
Suppose that $G = \cart{G}{1}{k}$ where $G_i$ is an abelian group.  Then the
sum in (\ref{innersum-eq}) is equal to
  \begin{equation} \label{normMJ} s\norm{M_JUO}^2, \end{equation}
where $M_J$ is the orthogonal projection of $\C^s$ on the subspace $V_J$, $U$
is the normalized character table of $G$, and $O$ is the vector of
multiplicities of the design $D$.
\end{proposition}

Our goal is now to show that the quantity (\ref{normMJ}) is independent of the
group structure of $G$.

A very useful way to describe $V_J$ is as follows.  Associate to $G_i$ the
Euclidean space $\C^{s_i}$, where the components of a vector $v$ are indexed by
the elements of $G_i$. Let $e^{(i)}_g$ be the unit vector in $\C^{s_i}$ having
a 1 in the $g$th place and zeros elsewhere, so that $e^{(i)}_1 = [1, 0, \cdots,
0]'$. Define the subspaces $V_i \subset \C^{s_i}$, $i = 1,\ldots, k$, by
setting\[ \begin{array}{rll}
V_i &= ({e^{(i) \, }_1})^{\perp}, & i \in J, \\
    &= \vspan(e^{(i)}_1), & i \notin J,
\end{array} \]
where orthocomplement ($^\perp$) and span are within $\C^{s_i}$.  Thus the
vectors of $V_i$ have a zero in the first position if $i \in J$ and zeros in
all the other positions if $i \notin J$.

Let $P_i$ be the projection of $\C^{s_i}$ onto $\vspan(e^{(i)}_1)$, $I_i$ the
identity matrix, and $Q_i = I_i - P_i$.

\begin{proposition} \label{Vi-prop}  With the above definitions, we have
  \begin{equation} \label{VJ-eq} V_J = V_1 \otimes \cdots \otimes V_k. \end{equation}
The orthogonal projection $M_J$ of $\C^s$ on $V_J$ is given by
  \begin{equation} \label{MJ-eq} M_J = M_1 \otimes \cdots \otimes M_k, \end{equation}
where $M_i$ is the orthogonal projection of $C^{s_i}$ on $V_i$.  We have
  \[\begin{array}{rll}
M_i &= Q_i, & i \in J, \\
    &= P_i, & i \notin J.
\end{array} \]
\end{proposition}
\begin{proof} The vectors in this tensor product are sums of vectors of the
form
  \[ v_1 \otimes \cdots \otimes v_k, \quad v_i \in V_i. \]
It is not hard to see that a vector of this form has zeros in exactly the
positions indexed by $g \notin S_J$, so that $V_1 \otimes \cdots \otimes V_k
\subset V_J.$  However,
  \[ \begin{array}{rll} \dim V_i &= s_i - 1,  & i \in J, \\ &= 1 &
  \mbox{otherwise,} \end{array} \]
so $\dim V_1 \otimes \cdots \otimes V_k = \dim V_1 \cdots \dim V_k = \prod_{i
\in J} (s_i-1) = \dim V_J$.  Thus (\ref{VJ-eq}) holds, and (\ref{MJ-eq}) follows immediately.  The formula for $M_i$ is obvious.
\end{proof}

We also note the following, which is implicit in equation (\ref{char_prod-eq}).
\begin{proposition} \label{char-table-product-prop}
If $G_i$ is a finite group having character table $H_i$ and normalized table
$U_i$, then $G = \cart{G}{1}{k}$ has character table $H = \ocart{H}{1}{k}$ and
normalized character table $U = \ocart{U}{1}{k}$.
\end{proposition}

We use this to evaluate the vector $M_JUO$ appearing in (\ref{normMJ}).  As in
(\ref{v=sum-eq}), the vector $O$ of multiplicities may be written
  \[ O = \sum_{g \in G} O(g) e_g. \]
But if $g = (g_1, \ldots, g_k) \in \cart{G}{1}{k}$, then
  \[ e_g = e^{(1)}_{g_1} \otimes \cdots \otimes e^{(k)}_{g_k} \]
where $e^{(i)}_j$ is the unit vector in $\C^{s_i}$ having a 1 in the $j$th
place and zeros elsewhere.  Then
  \[ O = \sum_{g = (g_1, \ldots, g_k)} O(g)\, e^{(1)}_{g_1} \otimes \cdots \otimes e^{(k)}_{g_k}. \]
Thus
  \begin{align} M_JUO&=
  \sum_{g = (g_1, \ldots, g_k)}O(g) \, M_J U(e^{(1)}_{g_1} \otimes \cdots \otimes e^{(k)}_{g_k}) \nonumber \\
  &=  \sum_{g = (g_1, \ldots, g_k)}O(g)\, M_1 U_1(e^{(1)}_{g_1}) \otimes \cdots \otimes M_k U_k(e^{(k)}_{g_k}).
  \label{expansion-eq} \end{align}
To analyze the (squared) norm of this, we need to analyze the terms in such
sums.  This leads to evaluating $M_iU_i$ on the basis elements $e^{(i)}_{g_i}$.
We will just need to do this when $M_i = P_i$.

\begin{lemma} \label{PUeg-lemma} For each $i$ and for every $g \in G_i$ we have
  \[ P_iU_i(e^{(i)}_{g}) = \frac{1}{\sqrt{s_i}} \, e^{(i)}_{1}
  . \]
\end{lemma}
\begin{proof}  For simplicity, suppress the index $i$.  Now for any $w \in \C^s$ we have
  \[  Pw = \inner{w,e_1}e_1, \]
so from Lemma \ref{unitary-lemma} we have
  \[  PUv = \inner{Uv,e_1}e_1 = \inner{v, U^*e_1}e_1 = \inner{v,b}e_1 \]
for any $v$, with $b = (1/\sqrt{s})[1, \ldots, 1]'$.  In particular,
  \[  PU(e_g) = \inner{e_g,b}e_1 = \frac{1}{\sqrt{s}} \, e_1
  , \] as claimed.  \end{proof}

We now evaluate the squared norm of sums of form (\ref{expansion-eq}). This
will rest on the following calculation.
\begin{lemma} \label{q=0-lemma} Let $I_j$ be identity matrix of order $s_j$, and let
\[ c_{g_1\cdots g_i} = \sum_{g_{i+1}, \ldots, g_k} O(g),  \]  the sum of the
numbers $O(g)$ over those $g \in G$ with the first $i$ values fixed at $(g_1,
\ldots, g_i)$.   Then for every $0 \le i \le k$,
  \begin{equation} \label{normcomp-eq}
  \norm{( \ocart{I}{1}{i} \otimes \ocart{P}{i+1}{k})UO}^2 =
  \frac{1}{s_{i+1}\cdots s_k} \sum_{g_1, \ldots, g_i} c_{g_1\cdots g_i}^2.
  \end{equation}
In particular, this quantity is independent of the group structure on $G$.
\end{lemma}
\begin{proof} We  see that  $(\ocart{I}{1}{i} \otimes
\ocart{P}{i+1}{k})UO$\begin{align*}
 &= \sum_{g = (g_1, \ldots, g_k)} O(g) \; I_1U_1(e^{(1)}_{g_1}) \otimes \cdots \otimes I_iU_i(e^{(i)}_{g_i})
 \otimes P_{i+1}U_{i+1}(e^{(i+1)}_{g_{i+1}}) \otimes \cdots \otimes P_kU_k(e^{(k)}_{g_k}) \\
 &= \sum_{g = (g_1, \ldots, g_k)} O(g) \quad U_1(e^{(1)}_{g_1}) \otimes \cdots \otimes
 U_i(e^{(i)}_{g_i}) \quad
 \otimes \quad \frac{1}{\sqrt{s_{i+1}}}e^{(i+1)}_1 \otimes \cdots \otimes \frac{1}{\sqrt{s_k}}e^{(k)}_1 \\
 &= \frac{1}{\sqrt{s_{i+1}\cdots s_k}} \sum_{(g_1, \ldots, g_i)} \left( \sum_{(g_{i+1}, \ldots, g_k)} O(g) \right)
 U_1(e^{(1)}_{g_1}) \otimes \cdots \otimes U_i(e^{(i)}_{g_i})  \otimes e^{(i+1)}_1 \otimes \cdots \otimes
 e^{(k)}_1 \\
 &= \frac{1}{\sqrt{s_{i+1}\cdots s_k}} \sum_{(g_1, \ldots, g_i)} c_{g_1\cdots g_i} \;
 U_1(e^{(1)}_{g_1}) \otimes \cdots \otimes U_i(e^{(i)}_{g_i})  \otimes e^{(i+1)}_1 \otimes \cdots \otimes
 e^{(k)}_1. \end{align*}
But for each $j$, the set $\{U_j(e^{(j)}_g), \; g \in G_j\}$ is orthonormal in
$\C^{s_j}$ as the unit vectors $e^{(j)}_g, g \in G_j$, are orthonormal and
$U_j$ is an isometry.  Hence the elements $U_1(e^{(1)}_{g_1}) \otimes \cdots
\otimes U_i(e^{(i)}_{g_i}) \otimes e^{(i+1)}_1 \otimes \cdots \otimes
e^{(k)}_1$ are orthonormal in $\C^s$, and so $\norm{( \ocart{I}{1}{i} \otimes
\ocart{P}{i+1}{k})UO}^2$ \begin{align*}  &=
 \frac{1}{s_{i+1}\cdots s_k} \sum_{(g_1, \ldots, g_i)} c_{g_1\cdots g_i}^2 \;
 \norm{U_1(e^{(1)}_{g_1}) \otimes \cdots \otimes U_i(e^{(i)}_{g_i})  \otimes e^{(i+1)}_1 \otimes \cdots \otimes
 e^{(k)}_1}^2.
\intertext{But this}
 &= \frac{1}{s_{i+1}\cdots s_k} \sum_{(g_1, \ldots, g_i)} c_{g_1\cdots g_i}^2 \;
 \norm{U_1(e^{(1)}_{g_1})}^2 \cdots \norm{U_i(e^{(i)}_{g_i})}^2 \; \norm{ e^{(i+1)}_1}^2 \cdots \norm{e^{(k)}_1}^2 \\
 &= \frac{1}{s_{i+1}\cdots s_k} \sum_{(g_1, \ldots, g_i)} c_{g_1\cdots g_i}^2
\end{align*}
as all the norms in the next-to-last line are 1.  This is formula
(\ref{normcomp-eq}).
\end{proof}

Now fix $J \subset \set{k}$.

\begin{proposition} \label{indep-prop} $\norm{M_JUO}^2$ is independent of
the group structure of $G$.
\end{proposition}

\begin{proof}  Actually, we will prove something more general, namely that the
proposition holds for projections $M$ made up of a tensor product of $P_i$'s,
$Q_i$'s and $I_i$'s, where $Q_i = I_i - P_i$.  Letting $q=$ the number of
factors $Q_i$ in the projection, we prove this by induction on $q$.

We simplify matters by proving our result for projections of form
 \begin{equation} \label{QIP-eq} M = \ocart{Q}{1}{q} \otimes \ocart{I}{q+1}{q+i} \otimes
 \ocart{P}{q+i+1}{k}.\end{equation}
The proof is the same for projections with other ordering of the tensor
factors.

The base case ($q=0$) is precisely Lemma \ref{q=0-lemma}.  For the induction
step, assume that the result holds for projections having $q-1$ factors $Q$
(not necessarily the first $q-1$ factors).  Now $Q_q = I_q - P_q$, so the
projection (\ref{QIP-eq}) is
  \begin{eqnarray*} \lefteqn{M = \ocart{Q}{1}{q-1} \otimes I_q \otimes \ocart{I}{q+1}{q+i} \otimes
 \ocart{P}{q+i+1}{k}} \\ & & - \ocart{Q}{1}{q-1} \otimes P_q \otimes \ocart{I}{q+1}{q+i} \otimes
 \ocart{P}{q+i+1}{k} \\ & & = T_1 - T_2, \end{eqnarray*}
say.  Since $T_1 = M + T_2$ and $M$ and $T_2$ are orthogonal, the Pythagorean
Theorem gives
  \begin{equation} \label{T1-T2-eq}
  \norm{MUO}^2 = \norm{T_1UO}^2 - \norm{T_2UO}^2. \end{equation}
But since $T_1$ and $T_2$ contain $q-1$ factors $Q_i$, the induction hypothesis applies to both terms on the right-hand-side of (\ref{T1-T2-eq}), and therefore to the left-hand-side, as desired.
\end{proof}

By Proposition \ref{MJUO-prop} this shows that the sum (\ref{innersum-eq}), and
therefore the quantities $A_j(D)$, are independent of the group structure of
$G$.  Theorem~\ref{indep-thm} is now proved.

\section{Conclusion} \label{conclusion}
The definition of the generalized wordlength pattern (GWLP) given in
\cite{XuWu01} makes sense if one chooses abelian rather than
cyclic groups to index the levels of each factor.  The choice to use cyclic groups
in \cite{XuWu01} is arbitrary, and we have shown that while it does affect the
so-called $J$-characteristics of a design, it does not affect the GWLP.  This
removes a possible ambiguity in the definition of the GWLP, and therefore in
the use of minimum aberration as an optimality criteria for nonregular designs.
The choice of cyclic groups may be useful computationally as the irreducible
characters are then especially simple.

A special case of the invariance with respect to group structure is already
implicit in the coding literature \cite{Delsarte73}.  (The connection with
regular designs is given in \cite{XuWu01}.) However, this covers designs in
which (a) the index sets $G_i$ are the same (the alphabet) and (b) the design
is actually a subset of $G$ (so that the counting function $O$ is simply an
indicator function). Our Theorem~\ref{indep-thm} is quite general, and makes no
use of concepts borrowed from coding theory.

The wordlength pattern of a regular design does not determine the design, and
in particular does not tell us its alias structure.  For that, one needs the
defining words.  We have seen that an analog of the set of defining words of a
nonregular design is the set of $J$-characteristics, at least in respect of
determining the design. However, as we noted in Section~\ref{J-char-sec}, the
$J$-characteristics vary with the choice of group structure assigned to
factors.  Certainly the aliasing structure of a design does not depend on this
arbitrary choice.  The GWLP is independent of this choice, and one may
therefore ask just what statistical information it carries.  This is a question
worthy of further investigation.

\vspace{1cm} \textbf{Acknowledgment}.  We thank Dan Lutter for some useful
discussions, and the referee and editor for some helpful suggestions.

\bibliography{paper}

\begin{thebibliography}{10}

\bibitem{AiZhang04}
Ming-Yao Ai and Run-Chu Zhang.
\newblock Projection justification of generalized minimum aberration for
  asymmetrical fractional factorial designs.
\newblock {\em Metrika}, 60:279--285, 2004.

\bibitem{BroidaWilliamson}
Joel~G. Broida and S.~Gill Williamson.
\newblock {\em A Comprehensive Introduction To Linear Algebra}.
\newblock Addison-Wesley Publishing Company Advanced Book Program, Redwood
  City, CA, 1989.

\bibitem{Delsarte73}
Philippe Delsarte.
\newblock Four fundamental parameters of a code and their combinatorial
  significance.
\newblock {\em Information and Control}, 23:407--438, 1973.

\bibitem{DengTang99}
Lih-Yuan Deng and Boxin Tang.
\newblock Generalized resolution and minimum aberration criteria for
  {Plackett--Burman} and other nonregular factorial designs.
\newblock {\em Statistica Sinica}, 9:1071--1082, 1999.

\bibitem{DieBed02}
Wiebke~S. Diestelkamp and Jay~H. Beder.
\newblock On the decomposition of orthogonal arrays.
\newblock {\em Utilitas Mathematica}, 61:65--86, 2002.

\bibitem{Franklin84}
M.~F. Franklin.
\newblock Constructing tables of minimum aberration $p^{n-m}$ designs.
\newblock {\em Technometrics}, 26:225--232, 1984.

\bibitem{FriesHunter80}
A.~Fries and W.~G. Hunter.
\newblock Minimum aberration $2^{k-p}$ designs.
\newblock {\em Technometrics}, 22:601--608, 1980.

\bibitem{Isaacs76}
I.~Martin Isaacs.
\newblock {\em Character Theory of Finite Groups}.
\newblock Academic Press, New York, 1976.

\bibitem{Ledermann87}
Walter Ledermann.
\newblock {\em Introduction to Group Characters}.
\newblock Cambridge University Press, Cambridge, 2nd edition, 1987.

\bibitem{Serre77}
Jean-Pierre Serre.
\newblock {\em Linear Representations of Finite Groups}.
\newblock Springer-Verlag, New York, 1977.
\newblock Leonard L. Scott, transl.

\bibitem{Takemura83}
Akimichi Takemura.
\newblock Tensor analysis of {ANOVA} decomposition.
\newblock {\em Journal of the American Statistical Association}, 78:894--900,
  1983.

\bibitem{Tang01}
Boxin Tang.
\newblock Theory of {$J$}-characteristics for fractional factorial designs and
  projection justification of minimum {$G_2$}-aberration.
\newblock {\em Biometrika}, 88:401--407, 2001.

\bibitem{TangDeng99}
Boxin Tang and Lih-Yuan Deng.
\newblock Minimum {$G_2$}-aberration criteria for nonregular fractional
  factorial designs.
\newblock {\em The Annals of Statistics}, 9:1914--1926, 1999.

\bibitem{Xu_personal}
Hongquan Xu.
\newblock Personal communication.

\bibitem{XuWu01}
Hongquan Xu and C.~F.~J. Wu.
\newblock Generalized minimum aberration for asymmetrical fractional factorial
  designs.
\newblock {\em The Annals of Statistics}, 29:1066--1077, 2001.

\end{thebibliography}
\bibliographystyle{plain}

\appendix
\section{Multilinear background}
In this section we briefly review some results on tensor products that we have
used in this paper.  We only deal with Euclidean spaces (specifically $\C^k$)
since that is all we need here.  For simplicity we concentrate on the bilinear
case (two tensor factors).

There are many expositions of multilinear algebra, such as that in
\cite{BroidaWilliamson}. An interesting exposition with some statistical
applications is given in \cite{Takemura83}.

As is well-known, the \emph{Kronecker} or \emph{tensor product} of the matrices
$A$ ($m \times n$) and $B$ is
  \[ A \otimes B = \left[
                     \begin{array}{ccc}
                       a_{11}B & \cdots & a_{1n}B \\
                       \vdots & \ddots & \vdots \\
                       a_{m1}B & \cdots & a_{mn}B \\
                     \end{array}
                   \right]. \]
For vectors $v \in \C^a$ and $w \in \C^b$  we thus have
  \[ v \otimes w = \left[
                     \begin{array}{c}
                       v_1 w \\
                       \vdots \\
                       v_a w \\
                     \end{array}
                   \right] \in \C^{ab}, \]
where $v = [v_1, \ldots, v_a]'$.  This product satisfies the usual bilinear
properties, for example, $cv \otimes w = c(v \otimes w)= v \otimes cw$ ($c$ a
scalar) and $A \otimes (B + C) = A \otimes B + A \otimes C$.

If $V \subset \C^a$ and $W \subset \C^b$ are subspaces, then we define their
tensor product to be the subspace of $\C^{ab}$ given by
  \[  V \otimes W = \vspan\{v \otimes w: \;  v \in V, w \in W\}. \]
(Technically, $V \otimes W$ is constructed as a free vector space modulo
bilinear relations, and is only isomorphic to a subspace of $\C^{ab}$, but
we will identify it with that subspace.)  If $\{e_1,
\cdots, e_k\}$ is a basis of $V$ and $\{f_1, \cdots, f_{\ell} \}$ is a basis of
$W$, then
  \[ \{ e_i \otimes f_j: \; i = 1, \ldots k, \; j = 1, \ldots \ell \} \]
is a basis of $V \otimes W$.  Thus in particular
  \[ \dim(V \otimes W) = \dim(V)\cdot \dim(W). \]

If we use $\inner{v_1,v_2}$ to denote the inner (or dot) product and $\norm{v}
= \sqrt{\inner{v,v}}$ the norm, then we have
  \begin{align*} \inner{v_1\otimes w_1, \, v_2 \otimes w_2} &= \inner{v_1,v_2}
  \inner{w_1,w_2},
\intertext{and in particular}
  \norm{v \otimes w} &= \norm{v}\norm{w}. \end{align*}
The norm and inner product of vectors of the form $\sum v_i \otimes w_i$ are
calculated by expanding the inner product in the usual way.

If $T_i$ is a linear transformation on $V_i$, then $T = T_1 \otimes T_2$ is a
linear transformation on $V_1 \otimes V_2$ such that
  \[ T(v_1 \otimes v_2) = T_1(v_1) \otimes  T_2(v_2). \]
$T$ is evaluated on sums of such terms by linearity. The matrix of $T$ is given
by the Kronecker product of the matrices $T_i$. Finally, if $S = S_1 \otimes
S_2$ is a linear transformation such that $S_iT_i$ is defined for each $i$,
then
  \[ ST = S_1T_1  \otimes S_2T_2. \]

All of the preceding extends in the obvious way to more than two tensor
factors.

\end{document}